\documentclass{article}
\usepackage[english]{babel}
\usepackage[letterpaper,top=2cm,bottom=2.5cm,left=3cm,right=3cm,marginparwidth=1.75cm]{geometry}
\usepackage{enumitem}
\usepackage{amsmath}
\usepackage{graphicx}
\usepackage{caption}
\usepackage{tikz}
\usepackage{subcaption}
\usepackage[colorlinks=true, allcolors=blue]{hyperref}

\usepackage{amsthm,amsfonts,bbm,bm}
\newtheorem{theorem}{Theorem}

\newtheorem{lemma}[theorem]{Lemma}

\usepackage{cite}

\newcommand{\cC}{\mathcal{C}}
\newcommand{\ie}{{i.e., }}
\newcommand{\eg}{{e.g., }}

\title{Modeling, Inference, and Prediction in Mobility-Based Compartmental Models for Epidemiology}
\author{Ning Jiang\footnote{Department of Mathematics and Statistics, University of Massachusetts Amherst, Amherst, Massachusetts, the United States, 01003}, \quad Weiqi Chu$^*$, \quad Yao Li$^*$}
\date{}

\begin{document}
\maketitle

\begin{abstract}

Classical compartmental models in epidemiology often assume a homogeneous population for simplicity, which neglects the inherent heterogeneity among individuals. This assumption frequently leads to inaccurate predictions when applied to real-world data. For example, evidence has shown that classical models overestimate the final pandemic size in the H1N1-2009 and COVID-19 outbreaks. To address this issue, we introduce individual mobility as a key factor in disease transmission and control. We characterize disease dynamics using mobility distribution functions for each compartment and propose a mobility-based compartmental model that incorporates population heterogeneity. Our results demonstrate that, for the same basic reproduction number, our mobility-based model predicts a smaller final pandemic size compared to the classical models, effectively addressing the common overestimation problem. Additionally, we infer mobility distributions from the time series of the infected population. We provide sufficient conditions for uniquely identifying the mobility distribution from a dataset and propose a machine-learning-based approach to learn mobility from both synthesized and real-world data.
\end{abstract}

\section{Introduction}
Understanding and predicting epidemic dynamics are essential for effective public health planning and response~\cite{shaman2012forecasting,holmdahl2020wrong,viboud2018rapidd}. Reliable and effective models provide crucial insights that enable timely implementation of control measures, efficient resource allocation, and strategic planning. For example, during the COVID-19 pandemic, models by Ferguson et al.~\cite{ferguson2020report} and IHME~\cite{covid2021modeling} guided decisions on school closures, social distancing, mask mandates, and reopening. Epidemiological models significantly influenced policy and played a key role in saving lives.

Among all these models, compartmental models have served as a cornerstone in epidemiology since the early 20th century. These compartmental models, including the well-known SIR (Susceptible-Infected-Recovered) model introduced by Kermack and McKendrick~\cite{kermack1927contribution} in 1927, divide the individuals into compartments based on their disease status and use ordinary differential equations to describe the dynamic evolution of each compartment.
Despite their utility, classical compartmental models have significant limitations~\cite{keeling2011modeling,anderson1991infectious}. A major issue is their assumption of a homogeneous mixing population, where each individual has an equal chance of interacting with any other. This overlooks spatial dynamics~\cite{ferguson2005strategies,yang2024epidemical}, real-world contact patterns~\cite{lloyd2005superspreading,nishiura2011did,chowell2007comparative}, and dynamic changes in human behavior during epidemics~\cite{farboodi2021internal,funk2010modelling}, such as increased handwashing, social distancing, and vaccination, which can significantly affect epidemic outcomes. The COVID-19 pandemic further highlights that classical models often ignore behavioral changes and interventions, leading to inflated estimates of susceptible populations.

To address these limitations, we propose mobility-based compartmental models incorporating heterogeneous mobility distributions. These models describe the time evolution of population densities for each compartment as a function of a mobility variable ranging from $[0, 1]$, making them infinite-dimensional. While the literature has explored similar infinite-dimensional models for their basic mathematical properties, to the best of our knowledge, this paper addresses critical topics in epidemiology such as the estimation of basic reproduction number and final pandemic size and the identifiability of the mobility distribution, which have not been previously studied.

A notable innovation in our paper is the introduction of individual mobility, a key factor in infectious disease spread. During an outbreak, individuals with higher mobility face greater risks of infection and transmission early on. As the pandemic progresses, this high-risk group recovers and indirectly protects the low-mobility group, thereby slowing disease spread and limiting the overall pandemic size. Classical ODE-based compartmental models cannot capture this mechanism. We prove that for the same basic reproduction number $\mathcal{R}_0$, a mobility-based SIR model predicts a smaller final pandemic size compared to a classical SIR model, addressing the common overestimation issue in the literature.

Another key contribution of this paper is the inference of mobility distribution from real-life data. For accurate forecasts and to match observed data, it is crucial to use a mobility distribution from real-world measurements. However, obtaining reliable mobility data is challenging due to large population sizes, privacy constraints, and inherent variability in individual movement patterns. 
We propose a method to infer the mobility distribution from accessible pandemic data, specifically the ratio of the infected population $I(t)$. We prove a one-to-one correspondence between the mobility distribution and the time series of $I(t)$, implying that the inference problem admits a unique solution. We also construct a training set of mobility distributions using Gaussian mixtures and employ an artificial-neural-network approach to identify the mobility distribution from real-life data.

Our paper proceeds as follows. In section \ref{sec: modeling}, we introduce a mobility-based compartmental model and emphasize its distinctions from homogeneous models, particularly in terms of the basic reproduction number and the final pandemic size. In section \ref{sec: learning}, we formulate the inverse problem to infer the mobility distribution from the time series of the infected population. We prove this inverse problem has a unique solution and propose a machine-learning approach to infer the mobility distribution. We conclude in section \ref{sec: conclusion}.

\section{Model description and epidemiological parameters} \label{sec: modeling}

\subsection{A mobility-based SIRS model}

The baseline SIRS model categorizes agents into three compartments: susceptible (S), infected (I), and recovered (R). Consider a system with $N$ agents, in which agent $i$ holds a disease state $X_i$ that takes a value of S, I, or R. We define $S(t)$, $I(t)$, and $R(t)$ as the ratios of agents in the three states:
\begin{equation}
S(t) = \frac1N\sum_{i = 1}^{N} \mathbf{1}_{ \{X_i(t) = \text{S}\} }, \quad I(t) = \frac1N\sum_{i = 1}^{N} \mathbf{1}_{ \{X_i(t) = \text{I}\} }, \quad R(t) = \frac1N\sum_{i = 1}^{N} \mathbf{1}_{ \{X_i(t) = \text{R}\}},
\end{equation}
where $\mathbf{1}$ is the indicator function.
The baseline SIRS model assumes individuals cycling between the S, I, and R states through infection, recovery, and loss of immunity. The baseline SIRS model without birth and death reads as
\begin{equation} \label{eq: SIRS}
    \begin{aligned}
    \dot{S}(t) &= -{\beta} S(t)I(t) + \varepsilon R(t), \\
    \dot{I}(t) &= {\beta} S(t)I(t) - \gamma I(t), \\
    \dot{R}(t) &= \gamma I(t) - \varepsilon R(t),
    \end{aligned}
\end{equation}
where $\beta$ is the infection rate, $\gamma$ is the recovery rate, and $\varepsilon$ is the loss of immunity rate.

Mobility is a key factor in epidemics, influencing the spread of infection across regions through travel, commuting, and migration. We let each agent have a time-independent mobility $m_i \in (0, 1)$, which determines their infection and transmission rates. For any S-I pair (\ie a susceptible agent $i$ and an infected agent $j$), the probability of infection is given by $\beta m_i m_j$.
Infected agents recover at a constant rate $\gamma$, and recovered agents lose immunity at a constant rate $\varepsilon$. Figure \ref{fig: SIR_diagram} illustrates the transitions between the three compartments and their respective rates.
\begin{figure}[htp]
    \centering
    \includegraphics[width=0.33\textwidth]{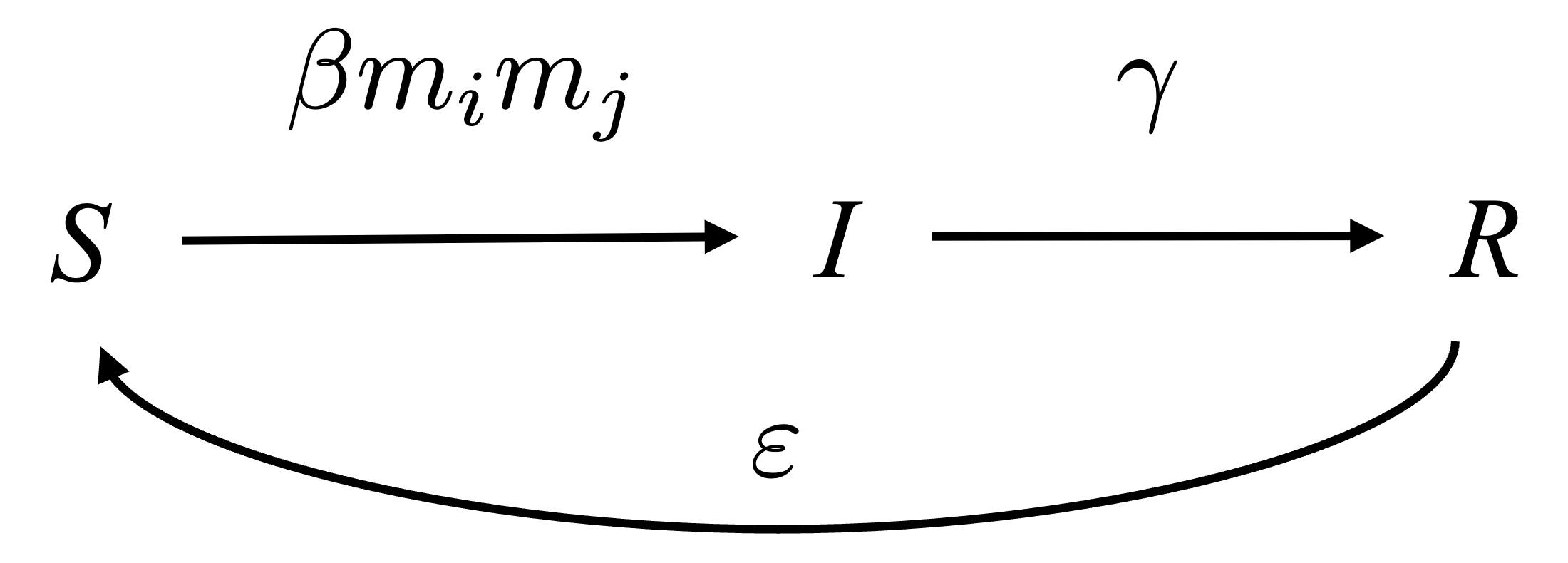}
    \caption{A schematic diagram to show the transition between three compartments and the transition rates in the mobility-based compartments models.}
    \label{fig: SIR_diagram}
\end{figure}

We describe the dynamics using three population density functions of mobility---$S(m,t)$, $I(m,t)$, $R(m,t)$---associated with the susceptible, infected, and recovered compartments, respectively. Here, $S(m,t)\mathrm{d}m$ represents the ratio of susceptible agents with mobility in the interval $[m,m+\mathrm{d}m)$. 
Let $f$ be the mobility distribution function of the whole population and we assume $f$ is time-independent. 
Due to population conservation, we have
\begin{equation} \label{eq: SIRf}
    S(m,t) + I(m,t) + R(m,t) = f(m), \quad \text{for all} ~ t \ge 0.
\end{equation}
While individuals with different mobility can interact in numerous ways, we assume that interactions between susceptible and infected individuals still follow the mass-action principle, as in the classical SIR model.
We assume that the probability of infection between an S-I pair is proportional to the product of their mobilities.
Therefore, the corresponding mobility-based SIRS model is
\begin{subequations} \label{eq: densitySIRS}
    \begin{align}
    \frac{\partial}{\partial t} S(m,t) &= -{\beta} \int_0^1 m\bar{m}S(m,t)I(\bar{m},t)~ \mathrm{d}\bar{m} +\varepsilon R(m,t), \\
    \frac{\partial}{\partial t} I(m,t) &= {\beta} \int_0^1 m\bar{m}S(m,t)I(\bar{m},t)~ \mathrm{d}\bar{m} - \gamma I(m,t), \label{eq: Im_eq} \\
    \frac{\partial}{\partial t} R(m,t) &= \gamma I(m,t) - \varepsilon R(m,t).
    \end{align}
\end{subequations}
In the mobility-based SIRS model, a susceptible agent with mobility $m$ can be infected by any infected agent with mobility $\bar{m}$, leading to an integral representation of the internal infection process. In our model, each agent's mobility is time-independent, resulting in local terms for the recovery and loss of immunity processes.

If all agents have the same mobility $m_0$, then the mobility distribution is a Dirac delta function centered at $m_0$ (\ie $f(m)=\delta(m-m_0)$). The mobility-based SIRS model \eqref{eq: densitySIRS} becomes the baseline SIRS model \eqref{eq: SIRS} with the infection rate $\beta_\text{eff}=\beta m_0^2$, where we refer $\beta_\text{eff}$ to as the effective infection rate.

\subsection{Dynamics of the mobility-based SIRS model}

In Figure \ref{fig: mobility-based SIRS-dynamics}, we illustrate the time evolution of the mobility-based SIRS model \eqref{eq: densitySIRS} during the first infection wave. Initially, the population is evenly distributed across the mobility domain $[0,1]$. As the simulation begins, susceptible agents with high mobility are infected first, leading to a peak in the infected population density near $m=1$ at $t=250$. Over time, the infected population recovers, resulting in an increased density in $R(m,t)$. This infection-recovery process occurs sooner in the high-mobility region compared to the low-mobility region.
\begin{figure}[hpt]
    \centering
    \begin{subfigure}{0.3\textwidth}
        \caption{$S(m,t)$}
        \includegraphics[width=\hsize]{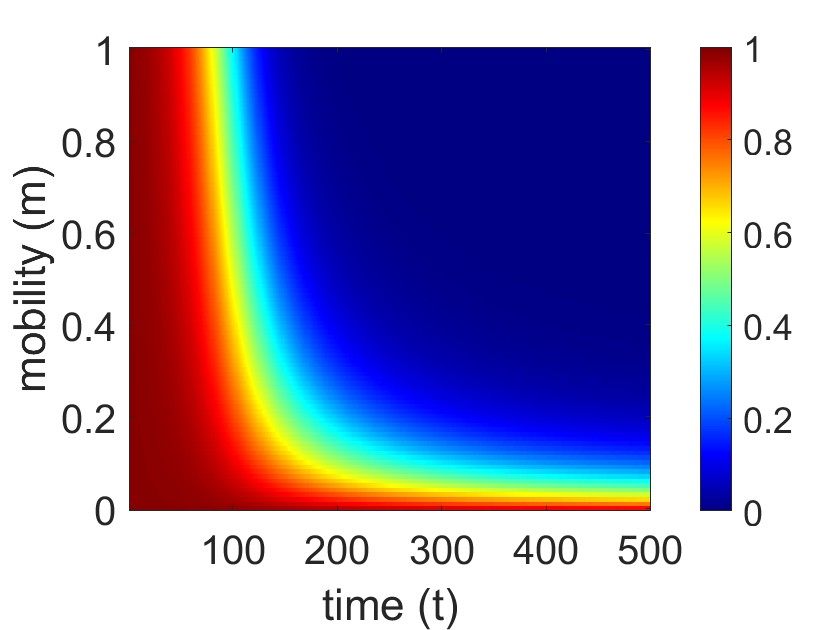}
    \end{subfigure}
    \hfill
    \begin{subfigure}{0.3\textwidth}
        \caption{$I(m,t)$}
        \includegraphics[width=\hsize]{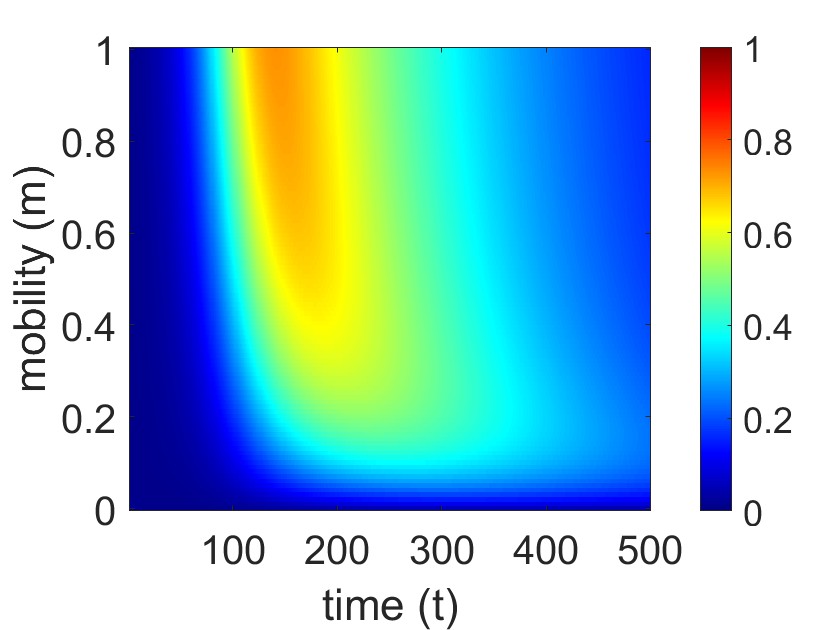}
    \end{subfigure}
    \hfill
    \begin{subfigure}{0.3\textwidth}
        \caption{$R(m,t)$}
        \includegraphics[width=\hsize]{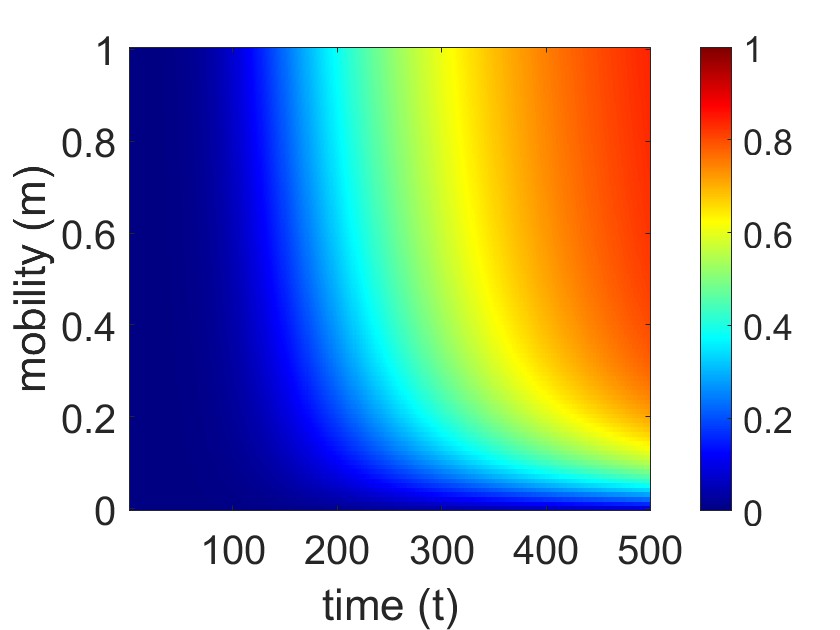}
    \end{subfigure}
    \caption{Evolution dynamics of the mobility-based SIRS model \eqref{eq: SIRf}, with parameter $\beta=0.1,\varepsilon=0.01$, $\gamma=0.3$ and initial $S(m,0)=0.99,I(m,0)=0.01,R(m,0)=0$ for (a) the susceptible population $S(m,t)$, (b) the infected population $I(m,t)$, and (c) the recovered population $R(m,t)$.}
    \label{fig: mobility-based SIRS-dynamics}
\end{figure}

The initial mobility distribution plays a crucial role in determining the mobility heterogeneity of the population as well as shaping the aggregated dynamics. We define the short-hand notation for the integral
\begin{equation}
    \langle g \rangle(t) := \int_0^1 g(m,t) ~\mathrm{d}m,
\end{equation}
and illustrate the time evolution of the ratios of three compartments in Figure \ref{fig: SIR-dynamics}. We use the initial conditions 
\begin{equation} \label{eq: homogeneous_initial}
S(m,0) = (1-I_0)f(m), \quad I(m,0) = I_0f(m), \quad R(m,0) = 0,
\end{equation} 
which we refer to as the proportional initial condition. 
To demonstrate the effect of heterogeneous mobility, we choose three initial mobility distributions $f_1$, $f_2$, and $f_3$, where $f_1$ and $f_2$ are restricted normal distributions on $[0,1]$ with 
\begin{equation}\label{eq: initial}
f_1 \propto 0.8\mathcal{N}(0.1, 0.02)+0.2\mathcal{N}(0.95,0.02), \quad  f_2 \propto \mathcal{N}(0.2,0.365),
\end{equation}
and $f_3$ is a Direct delta function $f_2(m)=\delta(m-0.435)$. All three mobility distributions have the same second-order moment $\langle m^2f\rangle$, which, as we will show in section \ref{sec: R0}, is proportional to the basic reproduction number. The mobility distribution $f_1$ is a mixture of two restricted Gaussian distributions and has two bumps centered at $m=0.1$ and $m=0.95$, indicating a polarized distribution of mobility among the population. The mobility distribution $f_3$ corresponds to the baseline SIRS model \eqref{eq: SIRS} in which everyone has the same mobility.

\begin{figure}[hpt]
    \centering
    \begin{subfigure}{0.3\textwidth}
        \caption{$\langle S\rangle(t)$}
        \includegraphics[width=\hsize]{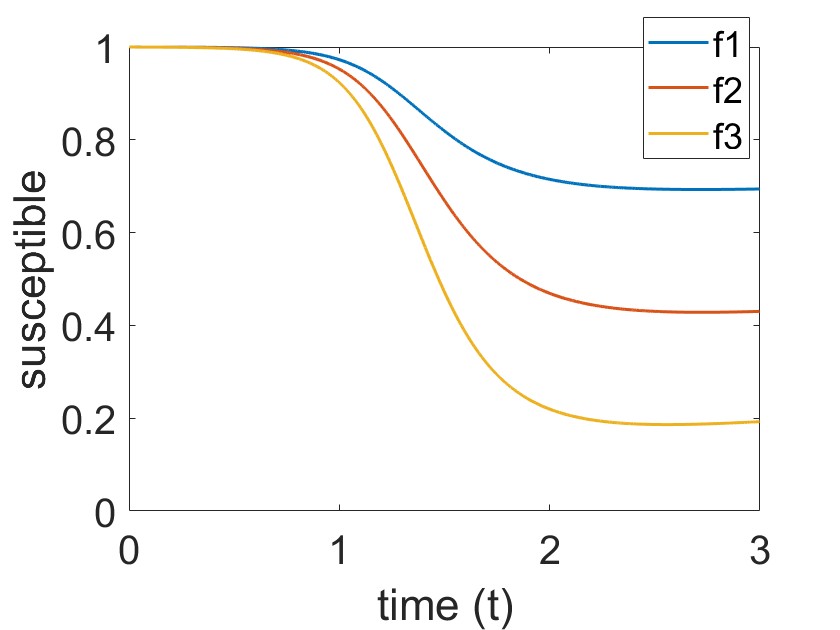}
    \end{subfigure}
    \hfill
    \begin{subfigure}{0.3\textwidth}
        \caption{$\langle I\rangle(t)$}
        \includegraphics[width=\hsize]{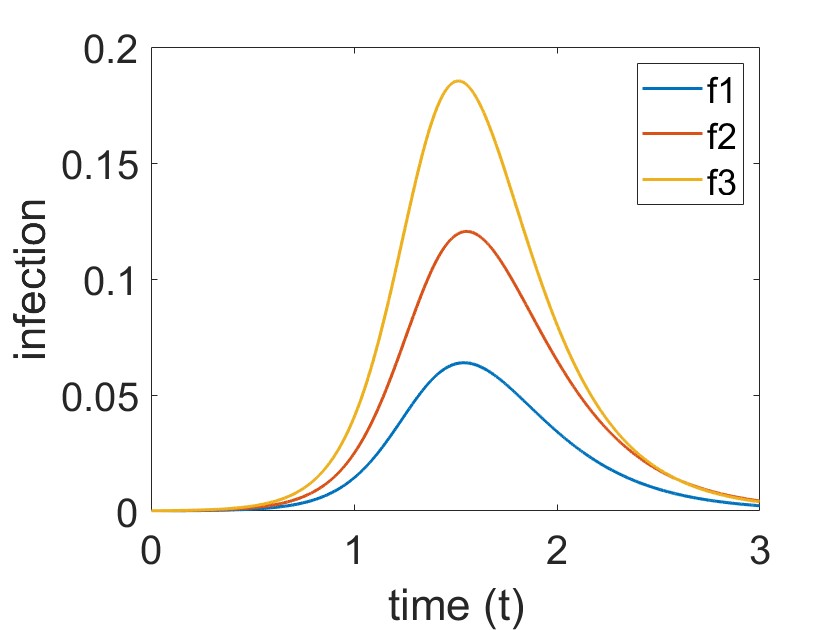}
    \end{subfigure}
    \hfill
    \begin{subfigure}{0.3\textwidth}
        \caption{$\langle R\rangle(t)$}
        \includegraphics[width=\hsize]{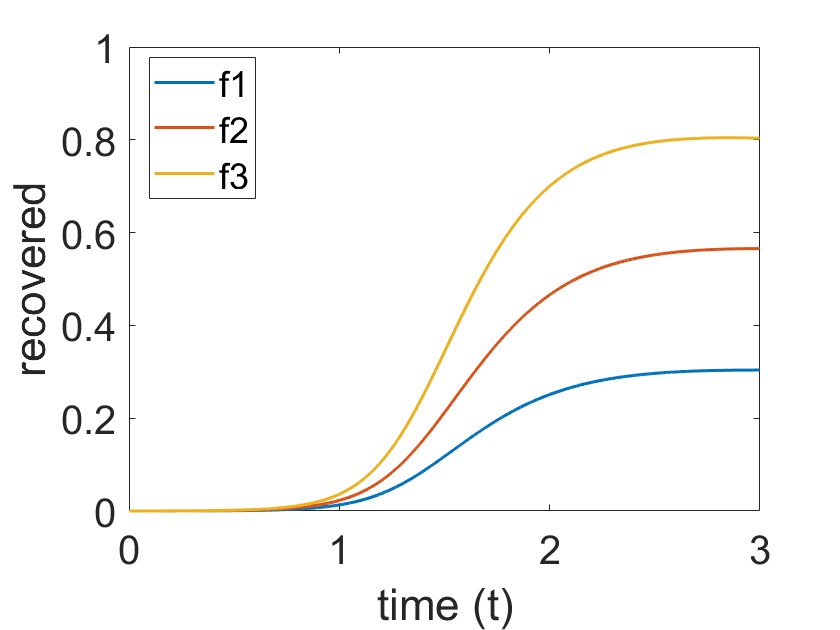}
    \end{subfigure}
    \caption{Time evolution of the ratios of three compartments in the mobility-based SIRS model \eqref{eq: SIRf}, with parameters $\beta=1.5$, $\gamma=0.13$, $\varepsilon=0.001$ and the proportional initial condition with $I_0=1\text{e-}4$ in \eqref{eq: homogeneous_initial}.}
    \label{fig: SIR-dynamics}
\end{figure}

We observe from Figure \ref{fig: SIR-dynamics} that at the end of the pandemic, over $70 \%$ of the population are infected for the homogeneous population case $f_3$. 
In contrast, $f_1$ has polarized mobility distribution with a relatively small group of high-mobile individuals. When high-mobile individuals recover from infection, they can provide protection to uninfected low-mobile groups. 
Figure \ref{fig: SIR-dynamics} shows that the final pandemic size regarding $f_1$ is only about $1/3$ of that of the baseline model. We derive and discuss the final pandemic size in detail in section \ref{sec: final_size}.

\subsection{Basic reproduction number} \label{sec: R0}

The basic reproduction number $\mathcal{R}_0$ is a key epidemiological metric that represents the average number of secondary infections produced by a single infected individual in a fully susceptible population. In this section, we derive the basic reproduction number $\mathcal{R}_0$ of the mobility-based SIRS model \eqref{eq: densitySIRS} by using the next-generation operator (\ie threshold operator) framework~\cite{thieme2011global}, which is essentially the infinite-dimensional version of the next generation matrix method~\cite{diekmann1990definition}. 

Recall the definition of $\mathcal{R}_0$ in \cite{thieme2011global}. For a generic infinite-dimensional SIRS model, suppose the infected population is defined on a compact set $\Omega$ and satisfies
\begin{equation}
\frac{\partial I(x,t)}{\partial t} = \int_\Omega h(x,y,S(x,t),I(y,t))~\mathrm{d}y - v(x) I(x,t), 
\end{equation}
where $h(x,y,S,I)$ is an incidence function and $v(x)$ is the recovery rate function. We define the threshold operator $\mathcal{L}$ as
\begin{equation}
\mathcal{L} u(x) = \int_\Omega \frac{u(y)}{v(x)} \left. \frac{\partial h(x,y,S,I)}{\partial I}\right|_{S=N(x),I=0} ~\mathrm{d}y,
\end{equation}
where $N(x)$ is the total population density, then the basic reproduction number $\mathcal{R}_0$ is the spectral radius of $\mathcal{L}$, which means
\begin{equation}
    \mathcal{R}_0 = \lim_{n \rightarrow \infty} \| (\mathcal{L})^n\|^{1/n}.
\end{equation}
We refer~\cite{thieme2011global} for detailed assumptions on $h$ and $v$.

For the mobility-based SIRS model defined in \eqref{eq: densitySIRS}, the total population density is $f(m)$ and the threshold operator $\mathcal{L}^f$ is 
\begin{align}
    \label{thresop}
    \mathcal{L}^f u(m) = \frac{\beta}{\gamma}m f(m)\int_0^1 \bar{m} u(\bar{m}) \mathrm{d}\bar{m} .
\end{align}
Therefore, the basic reproduction number $\mathcal{R}_0=\lim_{n \rightarrow \infty} \| (\mathcal{L}^f)^n\|^{1/n}$, and we compute $\mathcal{R}_0$ in Theorem \ref{prop: R0}.

\begin{theorem} \label{prop: R0}
    The basic reproduction number $\mathcal{R}_0$ of the mobility-based SIRS model \eqref{eq: densitySIRS} is 
    \begin{equation}
        \label{eq: R0}
            \mathcal{R}_0 = \frac{\beta}{\gamma}\langle m^2 f \rangle = \frac{\beta}{\gamma} \int_0^1 m^2 f(m) ~\mathrm{d}m .
    \end{equation}
\end{theorem}
\begin{proof}
Recall the definition of the threshold operator $\mathcal{L}^f$ in \eqref{thresop}. For any function $u(m)$, the image $\mathcal{L}^f u(m)$ lies in a 1-dimensional function space spanned by $mf(m)$. Consequently, $\mathcal{L}^f$ is a finite-rank operator, and thus, it must be compact. Therefore, the spectral radius of $\mathcal{L}^f$ is its maximal eigenvalue.

Because $\mathcal{L}^f$ is a rank-one operator, $\mathcal{L}^f$ only has 1 nonzero eigenvalue and this nonzero eigenvalue corresponds to the eigenfunction $mf(m)$. We compute this eigenvalue by applying $\mathcal{L}^f$ to $mf(m)$ and obtaining
\begin{equation}
\mathcal{L}^f( m f(m)) = \frac{\beta}{\gamma} m f(m)\int_0^1 \bar{m}^2 f( \bar{m}) \mathrm{d} \bar{m} = \frac{\beta}{\gamma}\langle m^2 f \rangle m f(m).
\end{equation}
Therefore, the leading eigenvalue of $\mathcal{L}^f$ is $\frac{\beta}{\gamma}\langle m^2 f \rangle$ and it is also the basic reproduction number $\mathcal{R}_0$. 

\end{proof}

The basic reproduction number, $\mathcal{R}_0$, is a key dimensionless epidemiological parameter that represents the average number of secondary cases generated by one primary case. When $\mathcal{R}_0 > 1$, the disease-free equilibrium becomes unstable, allowing the disease to persist in the population. Conversely, if $\mathcal{R}_0 < 1$, the disease-free equilibrium is asymptotically stable, leading to the eventual extinction of the disease.
The stability of both disease-free and endemic equilibria has been extensively studied over the past decades. For the original next-generation matrix method applied to finite-dimensional compartmental models, see~\cite{diekmann1990definition, diekmann2010construction}, and for various infinite-dimensional generalizations, refer to~\cite{thieme2009spectral, thieme2011global, qiu2018global}. By applying Theorems 3.16 and 3.17 in~\cite{thieme2009spectral}, one can readily show that the disease-free equilibrium $(f(m), 0, 0)$ of our model \eqref{eq: densitySIRS} is unstable when $\mathcal{R}_0 > 1$ and asymptotically stable when $\mathcal{R}_0 < 1$.

When the entire population has a constant mobility $m_0$, equivalent to the baseline SIRS model \eqref{eq: SIRS}, the basic reproduction number in \eqref{eq: R0} becomes
\begin{equation}
\mathcal{R}_0=\frac{\beta m_0^2}{\gamma}=\frac{\beta_\text{eff}}{\gamma},
\end{equation}
which is the same as the basic reproduction number of the baseline SIRS model \eqref{eq: SIRS} with the infection rate $\beta_\text{eff}$. 

Let $\text{var}(f)$ denote the variance of a mobility distribution $f$. In practice, we often estimate $\mathcal{R}_0$ by considering the early growth rate of cases and performing contact tracing studies, which is independent of models. 
Assuming that all distributions that have the same $\mathcal{R}_0$, from \eqref{eq: R0} we have
\begin{equation}
    \langle mf\rangle^2 = \frac{\gamma}{\beta}\mathcal{R}_0  - \text{var}(f) \le \frac{\gamma}{\beta}\mathcal{R}_0 = \langle mf_{\delta}\rangle^2.
\end{equation}
where $f_\delta$ is the Dirac delta distribution that has the same second-order moment as $f$. Therefore, the mean mobility $\langle mf_{\delta}\rangle$ of the Dirac delta distribution $f_\delta$ is larger than that of any other distribution $f$. 
Once $\mathcal{R}_0$ is determined first (using a model-free method), the homogeneous model (\eg the baseline SIRS model) overestimates the mean mobility, which consequently leads to an overestimation of the final pandemic size.

\subsection{Mean mobility}\label{sec: mean_mobility}
Mobility plays a crucial role in the mobility-based compartmental models. We focus on the mobility-based SIR model, which is defined by the following equations,
\begin{subequations} \label{eq: densitySIR}
    \begin{align}
    \frac{\partial}{\partial t} S(m,t) &= -{\beta} \int_0^1 m\bar{m}S(m,t)I(\bar{m},t) ~\mathrm{d}\bar{m} \label{eq: SIR_S}, \\
    \frac{\partial}{\partial t} I(m,t) &= {\beta} \int_0^1 m\bar{m}S(m,t)I(\bar{m},t)~ \mathrm{d}\bar{m} - \gamma I(m,t) \label{eq: SIR_I}, \\
    \frac{\partial}{\partial t} R(m,t) &= \gamma I(m,t), \label{eq: SIR_R} 
    \end{align}
\end{subequations}
and examine the mean mobility of the susceptible compartment. Let $Q_k^\varphi$ be the $k$th moment with respect to a nonnegative function $\varphi(m,t)$, which means 
\begin{equation}\label{eq: moments}
Q_k^\varphi(t)=\int_0^1 m^k\varphi(m,t) ~\mathrm{d}m.
\end{equation}
Under this definition, the ratio of the susceptible population is $Q_0^S(t)$, and the mean mobility of the susceptible compartment is $Q_1^S(t)/Q_0^S(t)$. 
We show the monotone property associated with the susceptible population in Theorem \ref{thm: moment_decrease}.

\begin{theorem}\label{thm: moment_decrease}
Let ${Q_k^S}(t)$ be the $k$th moment of $S(m,t)$ in the mobility-based SIR model \eqref{eq: densitySIR}. If $f$ is not a Dirac delta distribution, then for all $k\ge0$, the ratio ${Q_{k+1}^S}(t)/{Q_k^S}(t)$ strictly decreases with respect to time $t$.
\end{theorem}
\begin{proof}
   By multiplying $m^k$ to \eqref{eq: SIR_S} and integrating $m$ over $[0,1]$, we obtain that
   \begin{equation} \label{eq: dotQS}
       \dot{Q}_k^S = -\beta Q_1^IQ_{k+1}^S, \quad \text{for all} ~k\ge 0.
   \end{equation}
   Using \eqref{eq: dotQS}, we obtain the time derivative of ${Q_{k+1}^S}/{Q_k^S}$
   \begin{equation}\label{eq: dt_QS}
       \frac{\mathrm{d}}{\mathrm{d}t}\left(\frac{Q_{k+1}^S}{Q_k^S}\right) = \frac{1}{(Q_k^S)^2} \left(Q_k^S \dot{Q}_{k+1}^S-\dot{Q}_k^SQ_{k+1}^S\right) 
       = -\frac{\beta Q_1^I}{(Q_k^S)^2} \left[Q_k^SQ_{k+2}^S-(Q_{k+1}^S)^2\right].
   \end{equation}
   From H\"{o}lder's inequality, we have that 
   \begin{equation}
   Q_k^SQ_{k+2}^S-(Q_{k+1}^S)^2 > 0,
   \end{equation}
   which implies that the time derivative in \eqref{eq: dt_QS} is negative, so that ${Q_{k+1}^S}(t)/{Q_k^S}(t)$ strictly decreases over time $t$. 
\end{proof}

For $k=0$, the ratio in Theorem \ref{thm: moment_decrease} simplifies to $Q_1^S(t)/Q_0^S(t)$, representing the mean mobility of the susceptible population. Thus, in the mobility-based SIR model \eqref{eq: densitySIR}, the mean mobility of the susceptible population decreases over time $t$. However, this monotonic property does not necessarily apply to the infected and recovered populations, as illustrated in Figure \ref{fig: mean-mobility-dynamics}. In addition, if $f$ is a Dirac delta distribution, then the entire population has a constant mobility, causing the mean mobility to remain unchanged over time.
\begin{figure}[hpt]
    \centering
    \begin{subfigure}{0.3\textwidth}
        \caption{$Q_1^S/Q_0^S$}
        \includegraphics[width=\hsize]{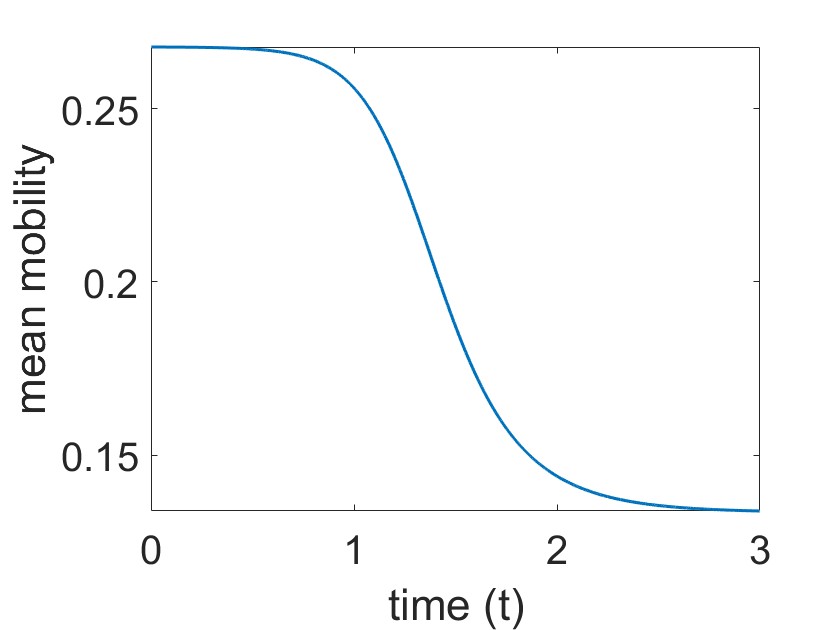}
    \end{subfigure}
    \hfill
    \begin{subfigure}{0.3\textwidth}
        \caption{$Q_1^I/Q_0^I$}
        \includegraphics[width=\hsize]{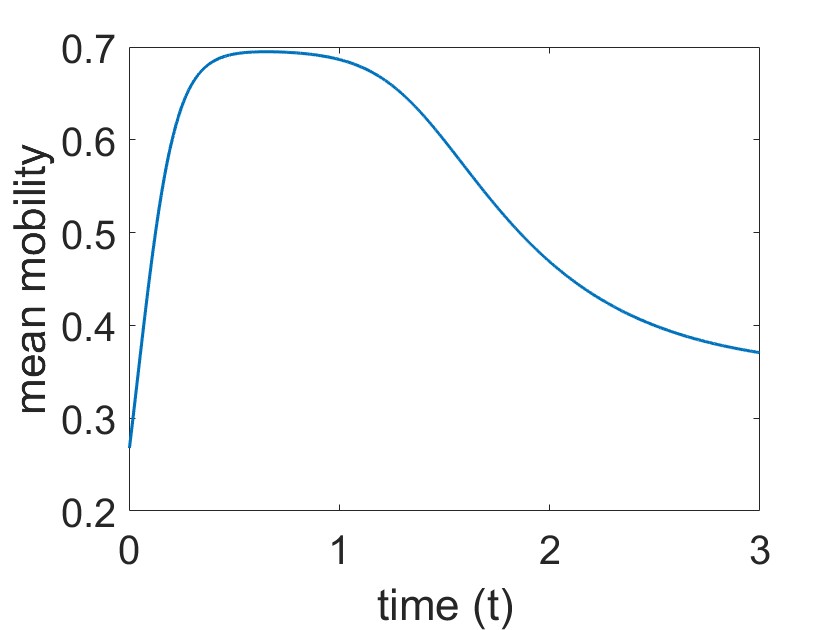}
    \end{subfigure}
    \hfill
    \begin{subfigure}{0.3\textwidth}
        \caption{$Q_1^R/Q_0^R$}
        \includegraphics[width=\hsize]{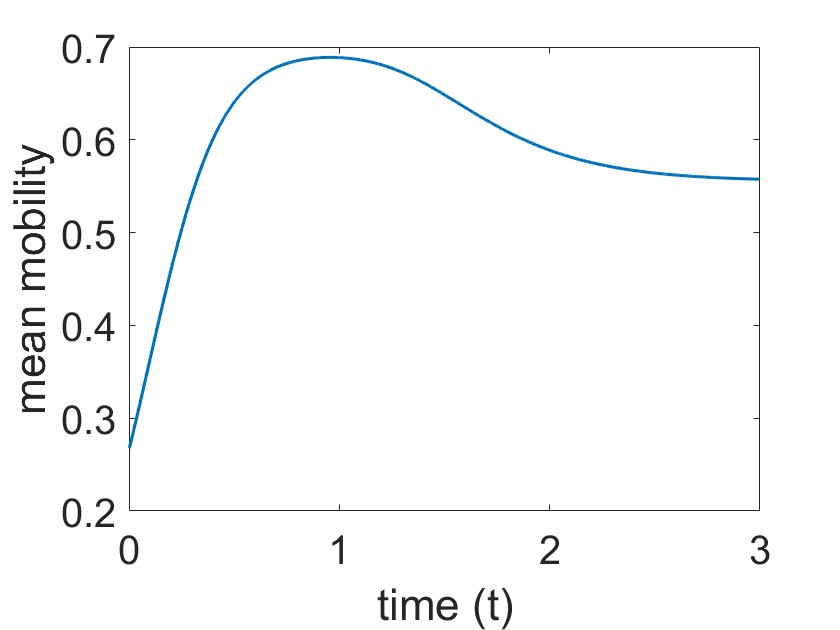}
    \end{subfigure}
    \caption{Time evolution of the mean mobility of three compartments in the mobility-based SIR model \eqref{eq: densitySIR} with parameters $\beta=1.5$, $\gamma=0.13$,  $I_0=0.0001$, and $f=f_1$ in \eqref{eq: initial}. The mean mobility of the susceptible population decreases strictly over time, while the mean mobility of the infected and the recovered populations increases first and then decreases as time progresses.}
    \label{fig: mean-mobility-dynamics}
\end{figure}

\subsection{Final pandemic size} \label{sec: final_size}
The final pandemic size, defined as the total number of individuals infected over the course of an epidemic, is a critical quantity for assessing disease transmissions. For the classical SIR model, it is well known that the final pandemic size $R_\infty$ satisfies
\begin{equation}
\label{eq: SIRfinalsize}
R_\infty + \exp( - \mathcal{R}_0 R_\infty) = 1,
\end{equation}
where $\mathcal{R}_0$ is the basic reproduction number.
Equation \eqref{eq: SIRfinalsize} often overestimates the final pandemic size. For instance, when $\mathcal{R}_0 = 2$, Equation \eqref{eq: SIRfinalsize} suggests that $R_\infty$ is around $0.8$, implying that nearly $80\%$ of the population will ultimately be infected, which is a scenario rarely observed in practice.

In reality, populations have heterogeneous mobility. High-mobility individuals are more likely to be infected early in a pandemic, and once they recover, they provide indirect protection to low-mobility individuals. Typically, $\mathcal{R}_0$ is estimated from the initial exponential growth rate, which is skewed by the higher infection rates among high-mobility individuals. This bias leads to an overestimation of the final pandemic size. 
Take the recent COVID-19 pandemic as an example. Studies estimate the basic reproduction number $\mathcal{R}_0$ for COVID-19 to be between 2 and 5~\cite{ke2021estimating,zhang2020estimation,d2020assessment} in early 2020. However, later variants, such as Delta and Omicron, prove significantly more contagious. Studies suggest that the $\mathcal{R}_0$ of the Omicron variant is as high as 10~\cite{wu2022emergence,liu2022effective,burki2022omicron}, which implies, according to the final pandemic size estimation in \eqref{eq: SIRfinalsize}, that nearly the entire population would eventually become infected. However, empirical data indicate otherwise. As of December 2021, after two years of multiple infection waves caused by various SARS-CoV-2 variants, the seroprevalence in the United States is only 33.5\%. Moreover, the largest infection wave, caused by the Omicron variant from December 2021 to February 2022, infects less than one-quarter of the total population~\cite{clarke2022seroprevalence}. These numbers are far below the final pandemic size suggested by the classical SIR model. 

In this section, we derive the final pandemic size of the mobility-based SIR model \eqref{eq: densitySIR} with the proportional initial condition \eqref{eq: homogeneous_initial}. We provide a quantitative explanation of why the classical SIR model tends to overestimate the final pandemic size in comparison to our mobility-based SIR model \eqref{eq: densitySIR}.

\begin{theorem}\label{thm: final_size_formula}
    Let $R_{\infty}$ be the final pandemic size of the mobility-based SIR model \eqref{eq: densitySIR} with the proportial initial condition \eqref{eq: homogeneous_initial}. Then the final pandemic size $R_{\infty}$ satisfies 
    \begin{equation} \label{eq: Rinfty}
    \begin{aligned}
        R_{\infty} + \langle f \exp\left(- \frac{\alpha m \beta}{\gamma}\right)\rangle &= 1, 
    \end{aligned}
    \end{equation}
    where $\alpha$ is the root on the interval $[0,1]$ of the equation
    \begin{equation}
    \alpha + \langle mf \exp\left(-\frac{\alpha m \beta}{\gamma}\right)\rangle = \langle mf \rangle.
    \end{equation}
\end{theorem}
\begin{proof}
We solve \eqref{eq: densitySIR} for $S(m,t)$ and obtain that
\begin{equation}\label{eq: Smt}
    S(m,t)=S(m,0)\exp\left(-\frac{\beta}{\gamma} m \langle mR\rangle (t)\right).
\end{equation}
We multiply the above equation by $m$ and integrate it over $m$, then obtain
\begin{equation}
    \langle mS \rangle(t) = \langle m S(m,0)\exp\left(-\frac{\beta}{\gamma} m \langle mR\rangle (t)\right)\rangle.
\end{equation}
Since $I(m,t)\to 0$ as $t\to\infty$, we have 
\begin{equation}
\langle mf \rangle - \langle mR \rangle(\infty) = \langle m S(m,0)\exp\left(-\frac{\beta}{\gamma} m \langle mR\rangle (\infty)\right)\rangle.
\end{equation}
Using the initial condition \eqref{eq: homogeneous_initial} and $S(m,0)\approx f(m)$, we find that $\alpha=\langle mR\rangle(\infty)$ satisfies
\begin{equation} \label{eq: alpha}
\alpha + \langle mf \exp\left(-\frac{\alpha m \beta}{\gamma}\right)\rangle = \langle mf \rangle.
\end{equation}
We integrate \eqref{eq: Smt} over $m$ and take $t$ to infinite, then we obtain that $R_\infty=\langle R\rangle(\infty)$ satisfies \eqref{eq: Rinfty}. 
\end{proof}

Theorem \ref{thm: final_size_formula} offers a method to compute the final pandemic size of the model \eqref{eq: densitySIR} efficiently. We first determine $\alpha$ by finding the root of \eqref{eq: alpha} within the interval $[0,1]$ and then use the value of $\alpha$ to compute $R_\infty$ using \eqref{eq: Rinfty}.
Once the parameters $\beta$ and $\gamma$ are fixed, the mobility distribution $f(m)$ solely determines both the final pandemic size $R_\infty$ and the basic reproduction number $\mathcal{R}_0$ as given in \eqref{eq: R0}. 
In Theorem \ref{thm: final-size-estimate}, we prove that among all distributions that yield the same basic reproduction number $\mathcal{R}_0$, the Dirac delta distribution $f_{\delta}$, which corresponds to a population with the same mobility, maximizes the final pandemic size.

\begin{theorem} \label{thm: final-size-estimate}
    Let $R_{\infty}$ and $R^{\delta}_{\infty}$ be the final pandemic sizes associated with the mobility distribution $f$ and the Dirac delta mobility distribution $f_\delta$, respectively.
    Suppose that $f$ and $f_{\delta}$ yield the same basic reproduction number $\mathcal{R}_0$ \eqref{eq: R0}, then $R_\infty \le R^{\delta}_{\infty}$ and the equality occurs only if $f=f_\delta$.
\end{theorem}
Theorem \ref{thm: final-size-estimate} elucidates why the classical SIR model tends to overestimate the final pandemic size compared to our mobility-based model. The classical SIR model assumes a homogeneous population, which corresponds to a Dirac delta mobility distribution \( f_\delta \) in the framework of our model \eqref{eq: densitySIR}. Its final pandemic size $R_\infty^\delta$ is greater than that of any other mobility distribution, given that all distributions share the same basic reproduction number $\mathcal{R}_0$.

We prepare the proof of Theorem \ref{thm: final-size-estimate} with the following lemma.

\begin{lemma} \label{lem: QIratio}
Let $I(m,t)$ be the mobility distribution of the infected population in the mobility-based SIR model \eqref{eq: densitySIR} with the proportional initial condition \eqref{eq: homogeneous_initial}. Then for all $t\ge 0$, it holds that
    \begin{equation} \label{eq: QI_inequality}
        \frac{Q_1^I(t)}{Q_0^I(t)} \le \frac{\langle m^2f \rangle}{\langle mf\rangle},
    \end{equation}
and the equality occurs only if $f$ is a Dirac delta distribution.
\end{lemma}
\begin{proof}
If $f$ is a Dirac delta distribution $f(m)=\delta(m-m_0)$, then both sides of \eqref{eq: QI_inequality} equal $m_0$. We assume that $f$ is not a Dirac delta distribution in the rest of the proof.

Using \eqref{eq: densitySIR}, we compute the time derivative of ${Q_1^I(t)}/{Q_0^I(t)}$ and obtain
\begin{equation}\label{eq: difference}
    \frac{\mathrm{d}}{\mathrm{d}t}\left(\frac{Q_1^I}{Q_0^I}\right) = \frac{\beta Q_1^IQ_1^S}{Q_0^I}\left(\frac{Q_2^S}{Q_1^S} - \frac{Q_1^I}{Q_0^I}\right).
\end{equation}
We denote the time that $\frac{\mathrm{d}}{\mathrm{d}t}\left({Q_1^I(t)}/{Q_0^I(t)}\right)$ changes its sign as $t_k$ and partition the time interval $(0,\infty)$ into open intervals $\Omega_k$ as shown in Figure \ref{fig: time-intervals}. 
\begin{figure}[htp]
\centering
\includegraphics[width=0.3\textwidth]{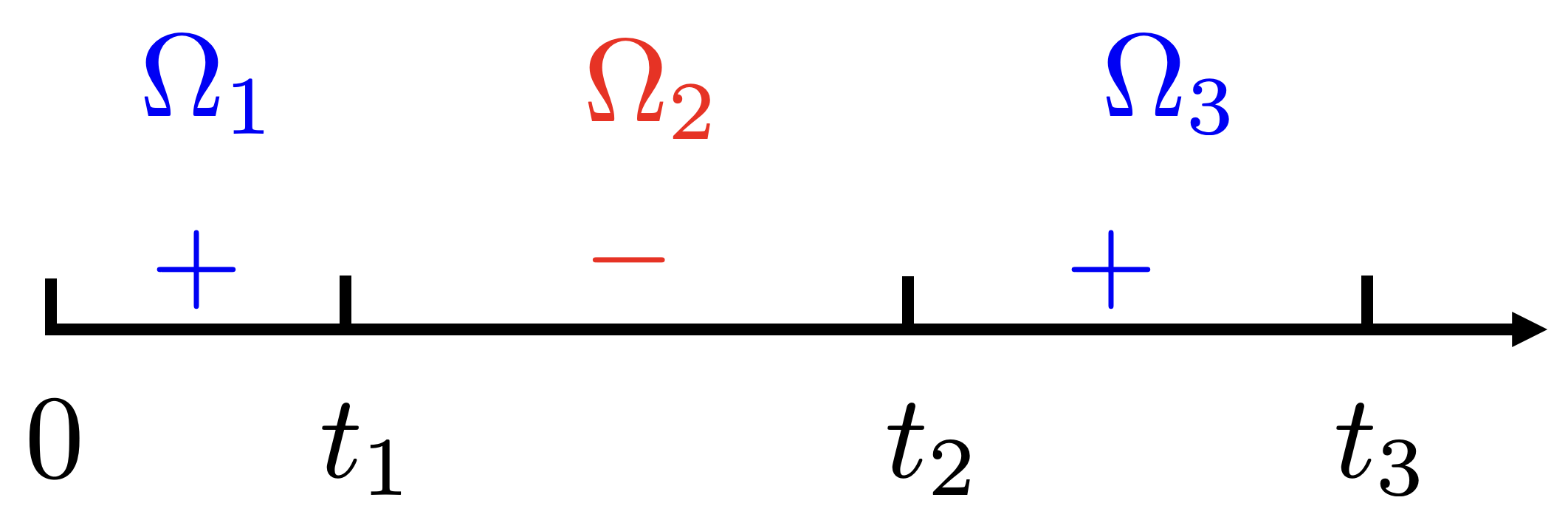}
\caption{A partition of intervals due to the sign of \eqref{eq: difference}.}
\label{fig: time-intervals}
\end{figure}

Using the initial condition \eqref{eq: homogeneous_initial}, we have
\begin{equation}
    \frac{\mathrm{d}}{\mathrm{d}t}\left(\frac{Q_1^I}{Q_0^I}\right)(0) = \frac{\beta Q_1^I(0)Q_1^S(0)}{Q_0^I(0)} \left(\frac{\langle m^2f\rangle}{\langle mf\rangle} - \langle mf\rangle \right) > 0,
\end{equation}
so $\frac{\mathrm{d}}{\mathrm{d}t}\left({Q_1^I}/{Q_0^I}\right)$ is positive on the interval $\Omega_1$. Due to the smoothness, the signs of \eqref{eq: difference} interlace on $\Omega_1$, $\Omega_2$, $\Omega_3$, $\ldots$.

If $t\in \cup_{k=1}^{\infty} {\Omega}_{2k-1}$ or $t\in\{t_1,t_2,\ldots\}$, then we know \eqref{eq: difference} is nonnegative, which implies that
\begin{equation} \label{eq: inequal1}
\frac{Q_1^I(t)}{Q_0^I(t)} \le \frac{Q_2^S(t)}{Q_1^S(t)} < \frac{Q_2^S(0)}{Q_1^S(0)} = \frac{\langle m^2f\rangle}{\langle mf\rangle},
\end{equation}
where the second inequality uses Theorem \ref{thm: moment_decrease} for the case $k=1$. 
If $t\in\cup_{k=1}^{\infty} \Omega_{2k}$ (supposing $t\in\Omega_{2k}$), then ${Q_1^I(t)}/{Q_0^I(t)}$ decreases strictly with respect to $t$ on the interval $\Omega_{2k}$. Due to continuity, we have
\begin{equation}\label{eq: inequal2}
\frac{Q_1^I(t)}{Q_0^I(t)} \le \frac{Q_1^I(t_{2k-1})}{Q_0^I(t_{2k-1})} < \frac{\langle m^2f \rangle}{\langle mf\rangle}, \quad \text{for all~} t\in\Omega_{2k-1},
\end{equation}
where the second inequality holds because of \eqref{eq: inequal1}. Equations \eqref{eq: inequal1} and \eqref{eq: inequal2} complete the proof.
\end{proof}

Now, we turn to the proof of Theorem \ref{thm: final-size-estimate}.
\begin{proof}
From \eqref{eq: densitySIR}, we have
\begin{equation}  \label{eq: QSR}
    \begin{aligned}
        \frac{\mathrm{d}Q_0^S}{\mathrm{d}t} = -\beta Q_1^SQ_1^I, \quad \frac{\mathrm{d} Q_0^R}{\mathrm{d}t} = \gamma Q_0^I, \quad \text{for all~} t\ge 0.
    \end{aligned}
\end{equation}
Since $Q_0^R(t)$ increases strictly with respect to $t$, we have a one-to-one correspondence between $Q_0^R$ and $t$. Using the chain rule, we have
\begin{equation}
\frac{\mathrm{d}Q_0^S}{\mathrm{d}Q_0^R}\frac{\mathrm{d}Q_0^R} {\mathrm{d}t} = -\beta Q_1^SQ_1^I,
\end{equation}
yielding
\begin{equation} \label{eq: dQ0S}
\frac{\mathrm{d}Q_0^S}{\mathrm{d}Q_0^R} = -\frac{\beta}{\gamma} \frac{Q_1^SQ_1^I}{Q_0^I}.
\end{equation}
Due to the one-to-one correspondence between $t$ and $Q_0^R$, we can view all functions of $t$ (including $Q_0^S$) as a function of $Q_0^R$ and rewrite \eqref{eq: dQ0S} as
\begin{equation} \label{eq: Qalpha}
    \frac{\mathrm{d}Q_0^S}{\mathrm{d} Q_0^R} = -\frac{\beta}{\gamma} \widetilde{H}(Q_0^R) Q_0^S, 
\end{equation}
where $\widetilde{H}$ is a function of $Q_0^R$ and defined by
\begin{equation} \label{eq: H}
\widetilde{H}(Q_0^R(t)) = \frac{Q_1^S(t) Q_1^I(t)}{Q_0^S(t)Q_0^I(t)}.
\end{equation}
Recall that $Q_0^R(0)=0$ and $Q_0^R(\infty)=R_\infty$. We integrate \eqref{eq: Qalpha} over the interval $[0,R_{\infty}]$ and use $Q_0^S(\infty)=1-R_\infty$, obtaining
\begin{equation} \label{eq: Ri}
R_\infty + \exp\left(-\frac{\beta}{\gamma} \int_0^{R_\infty} \widetilde{H}(\alpha) ~\mathrm{d}\alpha\right) = 1.
\end{equation}
Recall that $R_\infty^\delta$ satisfies \eqref{eq: SIRfinalsize}. We rewrite \eqref{eq: SIRfinalsize} as
\begin{equation}\label{eq: Ridelta}
R^\delta_\infty + \exp\left(-\int_0^{R_\infty^\delta} \mathcal{R}_0~ \mathrm{d}\alpha\right) = 1.
\end{equation}

From Theorem \ref{thm: moment_decrease} for $k=1$ and Lemma \ref{lem: QIratio}, we have
\begin{equation} \label{eq: two_inequal}
    \frac{Q_1^S(t)}{Q_0^S(t)}<\langle mf\rangle \quad \text{and} \quad  \frac{Q_1^I(t)}{Q_0^I(t)} < \frac{\langle m^2f \rangle}{\langle mf\rangle}, \quad \text{for all~} t>0.
\end{equation}
Using \eqref{eq: R0}, \eqref{eq: H}, and \eqref{eq: two_inequal}, we obtain that
\begin{equation}\label{eq: Halpha}
\frac{\beta}{\gamma}\widetilde{H}(\alpha) \le \mathcal{R}_0, \quad \text{for all~} \alpha \in [0,R_\infty),
\end{equation}
and the equality holds only if $f$ is a Dirac delta distribution (\ie $f=f_\delta$).

From \eqref{eq: Ri},\eqref{eq: Ridelta}, and \eqref{eq: Halpha}, we obtain that $R_\infty\le R_\infty^\delta$ and the equality holds only if $f=f_\delta$. 
\end{proof}

\section{Learning the mobility distribution from data} \label{sec: learning}
The mobility distribution plays a significant role in the mobility-based compartmental models. A realistic mobility distribution allows for more accurate predictions of disease transmission and enhances the precision of epidemiological models to reflect real-world scenarios. 
However, it is often impossible to measure the mobility of a whole population directly in real-world situations. 
In practice, the available data during epidemics is the number of infected individuals, which corresponds to the infected population ratio, defined as
\begin{equation} \label{eq: define_mI}
Q_0^I(t) = \int_0^1 I(m,t) ~ \mathrm{d}m,
\end{equation}
in the mobility-based SIRS model \eqref{eq: densitySIRS}.
In this section, we propose an inverse problem to infer the mobility distribution $f(m) \in \mathcal{L}^1(0,1)$ from the time-dependent infected ratio $Q_0^I(t)$. We prove that $f(m)$ is identifiable from $Q_0^I(t)$ (\ie the inverse problem has a unique solution) and provide a machine-learning-based approach to infer $f(m)$ from real data.

\subsection{Identifiability of the mobility distribution}
In this section, we show that the observation $Q_0^I(t)$ contains sufficient information to infer the mobility distribution $f(m)$. In other words, this inference problem has a unique solution. 

\begin{theorem}\label{thm: inverse}
Let $I(m,t)$ be the infected mobility distribution of the mobility-based SIRS model \eqref{eq: densitySIRS} with the initial condition \eqref{eq: homogeneous_initial} and $Q_0^I(t)$ be the 0th order moment of $I(m,t)$, defined in \eqref{eq: define_mI}. Assume that $Q_0^I(t)$ is arbitrarily smooth (\ie $Q_0^I(t)\in\cC^{\infty}(0,\infty)$) and that all the parameters ($I_0$, $\beta$, $\gamma$, and $\varepsilon$) are fixed and not equal to 0, then there is a one-to-one correspondence between $f$ and $Q_0^I$.
\end{theorem}

Recall the moment notation in \eqref{eq: moments}. We multiply $m^k$ to \eqref{eq: densitySIRS} and integrate over $m$, obtaining the equations for moments
\begin{subequations}\label{eq: momentsQ}
\begin{align} \label{eq: Qeq1}
    \frac{\mathrm{d}}{\mathrm{d}t} Q_k^S &= -\beta Q_1^IQ_{k+1}^S + \varepsilon Q_k^R, \\
        \frac{\mathrm{d}}{\mathrm{d}t} Q_k^I &= \beta Q_1^IQ_{k+1}^S - \gamma Q_k^I, \label{eq: Qeq2} \\
        \frac{\mathrm{d}}{\mathrm{d}t} Q_k^R &= \gamma Q_k^I -  \varepsilon Q_k^R. \label{eq: Qeq3}
\end{align}
\end{subequations}
Let $\tilde{S}(m,t), \tilde{I}(m,t), \tilde{R}(m,t)$ be the mobility distributions of \eqref{eq: densitySIRS} with the initial condition 
\begin{equation}
    \tilde{S}(m,0) = (1-I_0)\tilde{f}(m), \quad \tilde{I}(m,0) = I_0\tilde{f}(m), \quad \tilde{R}(m,0) = 0,
\end{equation}
and $Q_k^{\tilde S},Q_k^{\tilde I},Q_k^{\tilde R}$ be their $k$-th order moments. 

\begin{lemma}\label{lemma1}
    The following three statements are equivalent:
    \begin{equation}
    \begin{aligned}
    &Q_k^S(t)=Q_k^{\tilde{S}}(t), \quad &&\text{for all~} t\ge 0, \\
    &Q_k^I(t)=Q_k^{\tilde{I}}(t), \quad &&\text{for all~} t\ge 0, \\
    &Q_k^R(t)=Q_k^{\tilde{R}}(t), \quad &&\text{for all~} t\ge 0.
    \end{aligned}
    \end{equation}
\end{lemma}
\begin{proof}
    We add \eqref{eq: Qeq1} and \eqref{eq: Qeq2} together and obtain
    \begin{equation} \label{eq: derivative_Q}
    \frac{\mathrm{d}}{\mathrm{d}t}Q_k^S + \frac{\mathrm{d}}{\mathrm{d}t}Q_k^I = \varepsilon Q_k^f- (\varepsilon+\gamma)Q_k^I-\varepsilon Q_k^S,
    \end{equation}
    which uses the fact that $Q_k^f=Q_k^S+Q_k^I+Q_k^R$. Similarly, we have
    \begin{equation}
    \label{eq: derivative_Q_tilde}
    \frac{\mathrm{d}}{\mathrm{d}t}Q_k^{\tilde{S}} + \frac{\mathrm{d}}{\mathrm{d}t}Q_k^{\tilde{I}} = \varepsilon Q_k^{\tilde{f}} - (\varepsilon+\gamma)Q_k^{\tilde{I}} - \varepsilon Q_k^{\tilde{S}}.
    \end{equation}

    {\bf Step 1.} We show that $Q_k^I=Q_k^{\tilde{I}}$ is equivalent to $Q_k^S=Q_k^{\tilde{S}}$.
    Suppose that $Q_k^I=Q_k^{\tilde{I}}$. We compute the difference between \eqref{eq: derivative_Q} and \eqref{eq: derivative_Q_tilde} and 
    obtain
    \begin{equation} \label{eq: Qks_diff} 
    \frac{\mathrm{d}}{\mathrm{d}t}\left(Q_k^S - Q_k^{\tilde{S}}\right) = \varepsilon(Q_k^f-Q_k^{\tilde{f}})-\varepsilon(Q_k^S-Q_k^{\tilde{S}}).
    \end{equation} 
    Since $Q_k^I(0) = I_0Q_k^f$ and $Q_k^{\tilde{I}}(0) = I_0Q_k^{\tilde{f}}$, we use $Q_k^I(0)=Q_k^{\tilde{I}}(0)$ and obtain $Q_k^f = Q_k^{\tilde f}$, so that
    \begin{equation}
    Q_k^S(0) - Q_k^{\tilde{S}}(0) = (1-I_0)\left(Q_k^f-Q_k^{\tilde{f}}\right)=0.
    \end{equation}
    Therefore, $e(t) = Q_k^S(t) - Q_k^{\tilde{S}}(t)$ satisfies a linear ODE 
    \begin{equation}\frac{\mathrm{d}}{\mathrm{d}t} e(t)=-\varepsilon e(t), \quad e(0)=0,\end{equation}
    yielding $e(t)=0$ for all $t\ge 0$, so that $Q_k^S = Q_k^{\tilde{S}}$.
    Suppose that $Q_k^S=Q_k^{\tilde{S}}$. We use the similar idea as above and show $e(t)=Q_k^I(t) - Q_k^{\tilde{I}}(t)$ satisfies
    \begin{equation}
    \frac{\mathrm{d}}{\mathrm{d}t} e(t) = -(\varepsilon+\gamma)e(t),\quad e(0)=0,
    \end{equation}
    which yields $Q_k^I(t) = Q_k^{\tilde{I}}(t)$ for all $t\ge 0$.

    {\bf Step 2.} We then show that $Q_k^I=Q_k^{\tilde{I}}$ is equivalent to $Q_k^R=Q_k^{\tilde{R}}$.
    Suppose that $Q_k^R = Q_k^{\tilde{R}}$, then $\mathrm{d}Q_k^R/\mathrm{d}t = \mathrm{d}Q_k^{\tilde{R}}/\mathrm{d}t$. From \eqref{eq: Qeq3}, we have $Q_k^I = Q_k^{\tilde{I}}$. Suppose $Q_k^I = Q_k^{\tilde{I}}$. Using the similar idea as above, we have $e(t)=Q_k^R - Q_k^{\tilde{R}}$ satisfies
    \begin{equation} \frac{\mathrm{d}}{\mathrm{d}t} e(t) = -\varepsilon e(t),\quad e(0)=0,\end{equation}
    implying that $Q_k^R = Q_k^{\tilde{R}}$.
\end{proof}

Theorem \ref{thm: inverse} and Lemma \ref{lemma1} indicate that $f(m)$ also has a one-to-one correspondence to $Q_0^{S}(t)$ and $Q_0^{R}(t)$. This implies that, alternatively, we can infer the mobility distribution $f(m)$ from $Q_0^{S}(t)$ or $Q_0^{R}(t)$ if they are available, and the inference problems have unique solutions.

\begin{lemma}\label{lemma_derivative}
Let $L$ and $K$ be positive integers and $K\ge 2$. Assume the moments satisfy
    \begin{equation} \label{eq: lemma_eq}
    \begin{aligned}
    \frac{\mathrm{d}^\ell}{\mathrm{d}t^\ell} Q_k^{S,I,R}(0) &= \frac{\mathrm{d}^\ell}{\mathrm{d}t^\ell} Q_k^{\tilde{S},\tilde{I},\tilde{R}}(0), \quad k=0,1,\ldots,K-1-\ell, \\
    \frac{\mathrm{d}^\ell}{\mathrm{d}t^\ell} Q_k^{S,I,R}(0) &\neq \frac{\mathrm{d}^\ell}{\mathrm{d}t^\ell} Q_k^{\tilde{S},\tilde{I},\tilde{R}}(0), \quad k=K-\ell,
    \end{aligned}
    \end{equation}
for $\ell=0,1,\ldots,L-1$, then \eqref{eq: lemma_eq} also holds for $\ell=L$. 
\end{lemma}
\begin{proof}
A direct computation yields 
\begin{equation}
\frac{\mathrm{d}^{L}}{\mathrm{d}t^{L}} Q_k^I = \frac{\mathrm{d}^{L-1}}{\mathrm{d}t^{L-1}} (\beta Q_1^IQ_{k+1}^S - \gamma Q_k^I) = \beta\sum_{\alpha=0}^{{L-1}} C_{L-1}^\alpha \frac{\mathrm{d}^\alpha}{\mathrm{d}t^\alpha} Q_1^I \frac{\mathrm{d}^{{L-1}-\alpha}}{\mathrm{d}t^{{L-1}-\alpha}}Q_{k+1}^S -\gamma \frac{\mathrm{d}^{L-1}}{\mathrm{d}t^{L-1}} Q_k^I.
\end{equation}
For $k=0,\ldots,K-L-1$, we use conditions in \eqref{eq: lemma_eq} and obtain
\begin{equation}
\frac{\mathrm{d}^{L}}{\mathrm{d}t^{L}} Q_k^I(0)=\frac{\mathrm{d}^{L}}{\mathrm{d}t^{L}} Q_k^{\tilde I}(0).
\end{equation}
For $k=K-L$, we have
\begin{equation}\begin{aligned}
\frac{\mathrm{d}^{L}}{\mathrm{d}t^{L}} Q_{K-L}^I &= \beta\sum_{\alpha=0}^{{L-1}} C_{L-1}^\alpha \frac{\mathrm{d}^\alpha}{\mathrm{d}t^\alpha} Q_1^I \frac{\mathrm{d}^{{L-1}-\alpha}}{\mathrm{d}t^{{L-1}-\alpha}}Q_{{K-L}+1}^S -\gamma \frac{\mathrm{d}^{L-1}}{\mathrm{d}t^{L-1}} Q_{K-L}^I \\
&= \beta\sum_{\alpha=1}^{{L-1}} C_{L-1}^\alpha \frac{\mathrm{d}^\alpha}{\mathrm{d}t^\alpha} Q_1^I \frac{\mathrm{d}^{{L-1}-\alpha}}{\mathrm{d}t^{{L-1}-\alpha}}Q_{{K-L}+1}^S -\gamma \frac{\mathrm{d}^{L-1}}{\mathrm{d}t^{L-1}} Q_{K-L}^I 
+ \beta Q_1^I {\frac{\mathrm{d}^{L-1}}{\mathrm{d}t^{L-1}}Q_{{K-L}+1}^S}.
\end{aligned}\end{equation}
All paired moments from $f$ and $\tilde{f}$ are equal to each other at $t=0$ from \eqref{eq: momentsQ}, except for the last term since
\begin{equation}
{\frac{\mathrm{d}^{L-1}}{\mathrm{d}t^{L-1}}Q_{{K-L}+1}^S}(0)\neq {\frac{\mathrm{d}^{L-1}}{\mathrm{d}t^{L-1}}Q_{{K-L}+1}^{\tilde S}}(0).
\end{equation}
We complete the proof for moments for the infected population. We repeat similar calculations for the moments of $S$ and $R$ and obtain the results.
\end{proof}

Now we prove Theorem \ref{thm: inverse}.
\begin{proof}  
We use the same moment notation \eqref{eq: moments}. 
Assume that $f\neq \tilde{f}$, then we can find a positive integer $K$ such that
\begin{equation}
\begin{aligned}
&Q_k^f = Q_k^{\tilde f}, \quad \text{~for all}~ k=0,1,\ldots,K-1, \\
&Q_K^f \neq Q_K^{\tilde f}.
\end{aligned}
\end{equation}
From the proportional initial condition \eqref{eq: homogeneous_initial}, we have
\begin{equation}
\begin{aligned}
&Q_k^{S,I,R}(0) = Q_k^{\tilde{S},\tilde{I},\tilde{R}}(0), \quad \text{for all}~ k=0,1,\ldots, K-1,\\
&Q_K^{S,I,R}(0) \neq Q_K^{\tilde{S},\tilde{I},\tilde{R}}(0).
\end{aligned}
\end{equation}
We plug the initial values of moments into \eqref{eq: momentsQ} and obtain
\begin{equation} \label{eq: derivativeQ}
\begin{aligned}
&\frac{\mathrm{d}}{\mathrm{d}t} Q_k^{S,I,R}(0) = \frac{\mathrm{d}}{\mathrm{d}t} Q_k^{\tilde{S},\tilde{I},\tilde{R}}(0), \quad \text{for all}~ k=0,1,\ldots, K-2, \\
&\frac{\mathrm{d}}{\mathrm{d}t} Q_{K-1}^{S,I,R}(0) \neq \frac{\mathrm{d}}{\mathrm{d}t} Q_{K-1}^{\tilde{S},\tilde{I},\tilde{R}}(0),
\end{aligned}
\end{equation}
where the last inequality for the recovered population also uses the facts that
\begin{equation} Q_{K-1}^f=Q_{K-1}^S(0)+Q_{K-1}^I(0)+Q_{K-1}^R(0), \quad Q_{K-1}^{\tilde{f}}=Q_{K-1}^{\tilde{S}}(0)+Q_{K-1}^{\tilde{I}}(0)+Q_{K-1}^{\tilde{R}}(0).\end{equation}
Notice that \eqref{eq: derivativeQ} satisfies the conditions in Lemma \ref{lemma_derivative} for $L=2$. By induction, we know \eqref{eq: lemma_eq} hold for $\ell=2,3,\ldots,K$. In particular, we use the inequality result in Lemma \ref{lemma_derivative} for $\ell=K$ and $k=0$ and obtain
\begin{equation}
\frac{\mathrm{d}^K}{\mathrm{d}t^K} Q_0^{I}(0) \neq \frac{\mathrm{d}^K}{\mathrm{d}t^K} Q_0^{\tilde{I}}(0).
\end{equation}
Therefore, we have $Q_0^I\neq Q_0^{\tilde{I}}$ and the correspondence between $f$ and $Q_0^I$ is one-to-one.

\end{proof}

\subsection{Numerical inference of the mobility distribution}\label{sec: NNs}
Theorem \ref{thm: inverse} establishes a one-to-one correspondence between the mobility distribution $f(m)$ and the infected population ratio $Q_0^I(t)$. This provides a theoretical basis for inferring the mobility distribution from time series data of the infected population. One can infer the mobility distribution from early-stage pandemic data and use it to predict the future course of the pandemic. 
In this section, we propose a numerical approach using neural networks to achieve this inference.

\subsubsection{A neural-network-based learning framework}
Let $\langle I \rangle(t)$ be the ratio of the infected population at time $t$ (\ie $\langle I \rangle(t) = \int_0^1 I(m,t)~\mathrm{d}m$) and $\mathbf{I}$ be the observed time series of $\langle I \rangle(t)$, which takes the form of
\begin{equation}\label{eq: Idiscrete}
\mathbf{I} = (\langle I\rangle(t_1), \langle I\rangle(t_2), \ldots, \langle I\rangle(t_{n_t}))^T.
\end{equation}
Since the data is a finite-dimensional vector, we are only able to infer the mobility distribution $f$ in a finite-dimensional space. 
Let $\bar{f}$ be the spatial discretization of $f$ with $M$ uniform grids on the interval $[0,1]$. We take the discrete cosine transform (DCT) of $\bar{f}$ and denote the DCT of $\bar{f}$ by $\hat{f} = \sqrt{M}(1, \hat{f}_1,\ldots,\hat{f}_{n_f}, \ldots )$. 
Let $\mathbf{f}$ be the first $n_f$ nontrivial modes of $\hat{f}$, \ie
\begin{equation} \label{eq: fdiscrete}
\mathbf{f} = (\hat{f}_1,\hat{f}_2,\ldots,\hat{f}_{n_f})^T. 
\end{equation}
We formulate the inference problem using a feedforward neural network $\mathrm{G}_{\theta}$ to learn the mapping from $\mathbf{I}$ to $\mathbf{f}$:
\begin{equation}
  \mathbf{f} = \mathrm{G}_{\theta}( \mathbf{I}).
\end{equation}

To build the training set, we sample various mobility distributions $f(m)$ and simulate the mobility-based SIRS model \eqref{eq: densitySIRS} to generate the corresponding time series of $\langle I\rangle(t)$. 
When using neural networks to learn unknown parameters or functions, it's essential to sample a training set that covers a wide range of scenarios, ensuring the model is effective across the prediction regime.
In our context, this involves sampling the mobility distribution comprehensively and robustly. 
Since the Gaussian mixture model is a universal approximator of probability densities~\cite{goodfellow2016deep}, we sample $f$ as a Gaussian mixture model with random parameters. In particular, we let
\begin{equation}
{f} \propto \sum_{i = 1}^{\alpha} w_{i} \mathcal{N}( \mu_{i}, \sigma_{i}^{2})
\end{equation}
where $\mathcal{N}(\mu, \sigma^2)$ is the Gaussian probability density function and $\mu_i$, $\sigma_i$, and $w_i$ are random numbers.
Suppose the dataset contains $K$ pairs of $\mathbf{I}_k$ and $\mathbf{f}_k$, which we discretize from $\langle I_k\rangle(t)$ and $f_k(m)$. Then, we determine the fully connected feedforward neural network $\mathrm{G}_{\theta}$ by minimizing the loss function $\mathrm{L}(\theta)$:
\begin{equation} \label{eq: L}
\mathrm{L}(\theta) = \frac{1}{K}\sum_{k=1}^K \| \mathrm{G}_{\theta}( \mathbf{I}_k)-\mathbf{f}_k\|^2,
\end{equation}
where $\|\cdot\|$ is the $L_2$ vector norm.

\subsubsection{Numerical results}
In our simulations, we let $f_k$ be a Gaussian mixture of 3 normal distributions:
\begin{equation}\label{eq: f}
  {f_k} \propto \sum_{i = 1}^{3} w_{i} \mathcal{N}( \mu_{i}, \sigma_{i}^{2}),
\end{equation}
and sample the parameters $w_i$, $\mu_{i}$, and $\sigma_{i}$ as random numbers uniformly on intervals $[0, 1]$, $[-0.1, 1.1]$, and $[0.05, 1.05]$, respectively. 
We solve the mobility-based SIRS model \eqref{eq: densitySIRS} for $\langle I\rangle(t)$ with the initial condition \eqref{eq: homogeneous_initial} and parameters $\beta=0.8,\varepsilon=0.005$, $\gamma=0.14$, and $I_0=0.0001$. 
We sample $K = 10,000$ mobility distributions $f_k$ in the form of \eqref{eq: f} and compute corresponding $\langle I\rangle(t)$.
We discretize $f_k$ and $\langle I\rangle(t)$ with $n_t=101$ in \eqref{eq: Idiscrete} and $n_f=8$ in \eqref{eq: fdiscrete} and obtain the training set containing $K = 10,000$ input-output pairs $(\mathbf{I}_k,\mathbf{f}_k)$. 
To improve the performance, we scale the input data $\mathbf{I}_k$ by a factor of $50$ and the output data $\mathbf{f}_k$ by a factor of $3$. The mean values of the input and output sets are around $1$ after rescaling. 

In our simulations, we use a six-layer fully connected neural network, and the width of layers are 128, 512, 512, 512, 128, and 32, respectively. Each layer uses a Sigmoid activation function and an $L^2$ regularization of $\theta$ with a magnitude of $0.005$. 
For each training step, we sample a random batch from the training set and use the Adam optimizer with a learning rate of 0.0005 to minimize the loss function for each batch. We train the neural network for 2,000 epochs, and during each training epoch, we optimize the loss function using 1,562 batches with a batch size of 54.

We obtain the mapping $\mathrm{G}_\theta$ by minimizing the loss function \eqref{eq: L}. To test the performance of the learned mapping $\mathrm{G}_\theta$, we generate 6 additional mobility distributions ($\mathbf{f}_1, \ldots, \mathbf{f}_6$) using \eqref{eq: f} and simulate the corresponding time series of infected populations ($\mathbf{I}_1, \ldots, \mathbf{I}_6$) as inputs. 
We compare the predicted outputs $\hat{\mathbf{f}}_i = \mathrm{G}_\theta(\mathbf{I}_i)$ for $i = 1, \ldots, 6$ with the ground truth $\mathbf{f}_i$. We compute their inverse DCTs and show the comparison between the learned mobility distributions $\hat{f}_i$ and the actual distributions ${f}_i$ in Figure \ref{fig: learned-mobility}. The learned distributions successfully capture the overall trends of the true mobility distributions.
\begin{figure}[ht]
\begin{center}
    \includegraphics[width=0.85\textwidth]{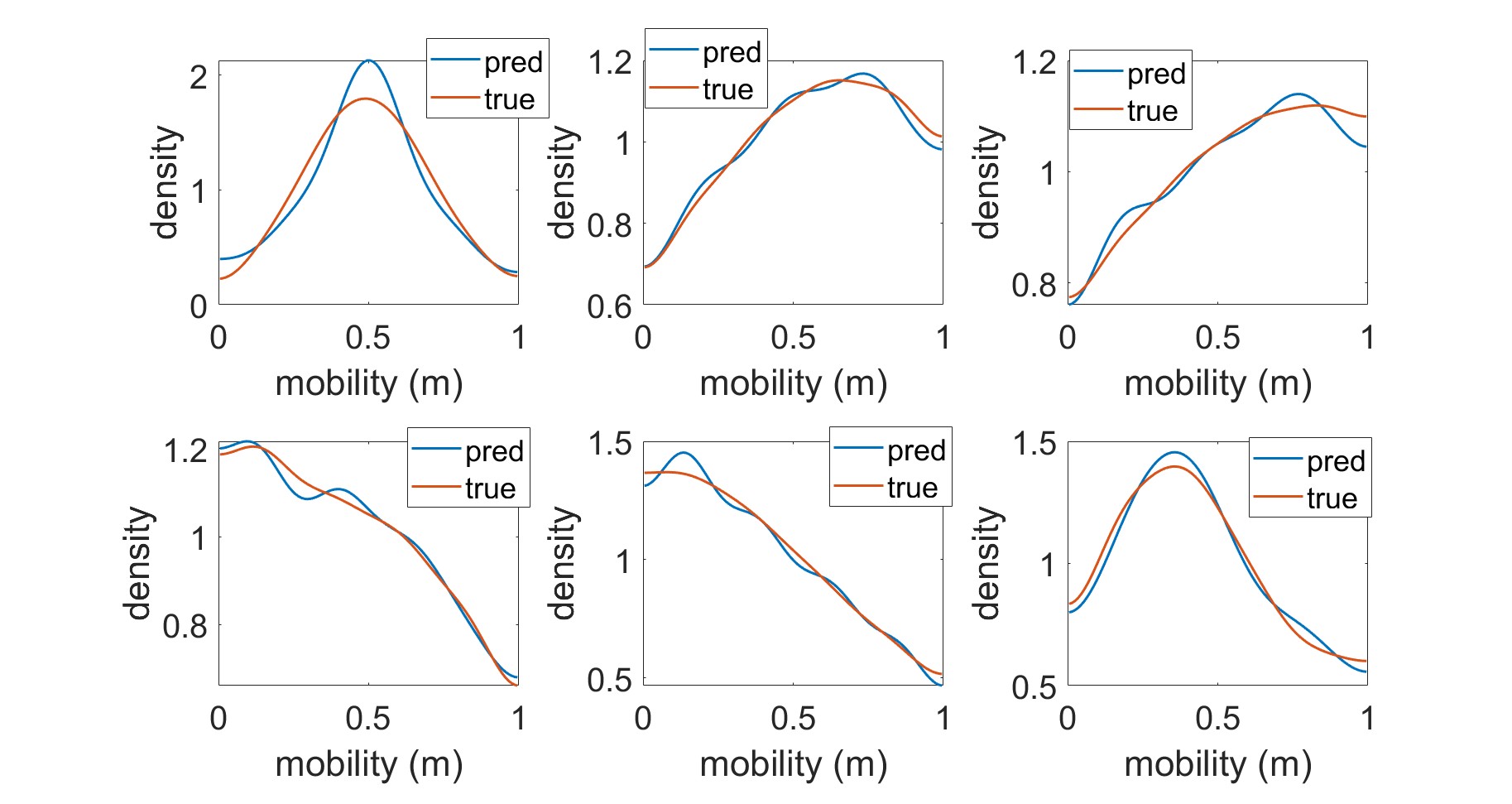}
\end{center}
\caption{A comparison of 6 predicted mobility densities $\hat{f}_i(m)$ and their ground truth $f_i(m)$.}
\label{fig: learned-mobility}
\end{figure}

We solve the mobility-based SIRS model \eqref{eq: densitySIRS} using the predicted mobility $\hat{{f}}_i$ and obtain the corresponding time series of the infected population $\langle \hat{I}_i\rangle(t)$. Figure \ref{fig: time_series} compares these predicted time series with the ground truth. The results show that the inferred mobility distributions generate time series of the infected population that closely match the ground truth.
\begin{figure}[ht]
    \centering
    \includegraphics[width=0.85\textwidth]{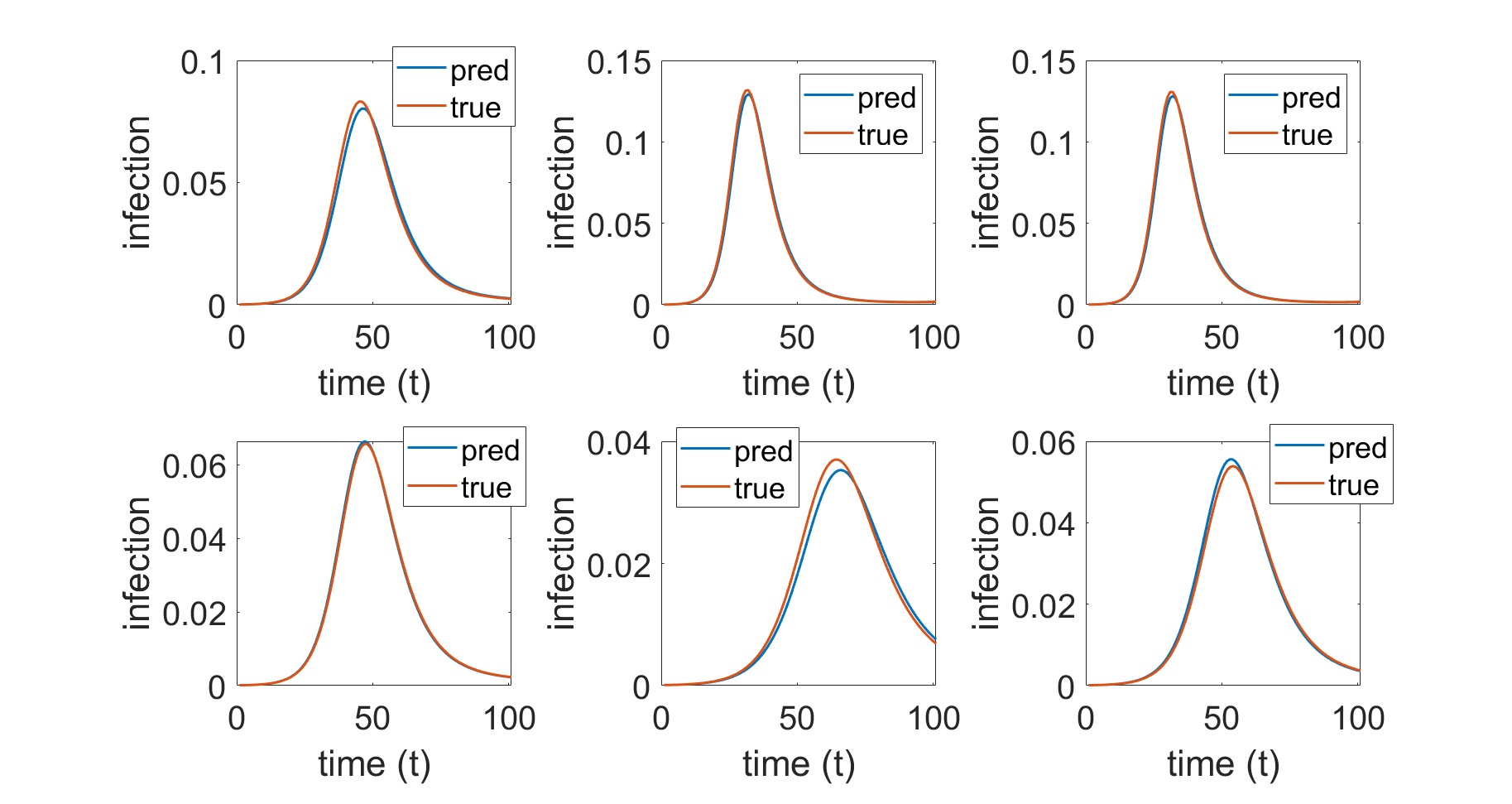}
    \caption{A comparison of the predicted infected time series $\hat{\langle I_i\rangle}(t)$ and their ground truth $\langle I_i\rangle(t)$.}
    \label{fig: time_series}
\end{figure}

Next, we focus on inferring mobility distributions using real-world COVID-19 data~\cite{jiang2024artificial}. The original dataset includes daily new infection counts for Massachusetts and New York from February 29th to September 16th, 2020. We calculate the total number of infected cases by summing the daily new infections with the previously infected cases, adjusted by a decaying factor to account for the recovered cases. In particular, we assume that the infected population $\langle I\rangle (t)$ follows 
\begin{equation} \label{eq: cases}
\langle I\rangle (t) =I_\text{daily}(t)+\exp{(-\gamma)} \langle I\rangle (t-1) ,
\end{equation}
where $I_{daily}(t)$ is the daily infection cases and $\gamma = 0.1345$ is the recovery rate, obtained from fitting an SIR model with the Massachusetts data obtained in~\cite{jiang2024artificial}. 
We scale the input $\langle I \rangle (t)$ by a factor of $50$, which is similar to the treatment of the training set. We show the rescaled time series of the COVID-19 infected population in Figure \ref{fig: real-data-learned}, together with the inferred mobility distributions from this dataset. 

The real-world data is influenced by numerous factors that have drastically altered social mobility during this period, such as social distancing and lockdown policies. Although our model does not account for these factors, the inferred mobility distribution still provides valuable insights. We observe that the inferred distributions include a small portion of high-mobility individuals, which accelerates disease transmission during the early stages of the COVID-19 pandemic. Classical models, which assume a homogeneous population, often fail to distinguish these high-mobility groups, leading to an overestimation of the pandemic size.

\begin{figure}[hpt]
    \centering
    \begin{subfigure}{0.36\textwidth}
        \caption{Real COVID-19 data $\langle I\rangle(t)$}
        \includegraphics[width=\hsize]{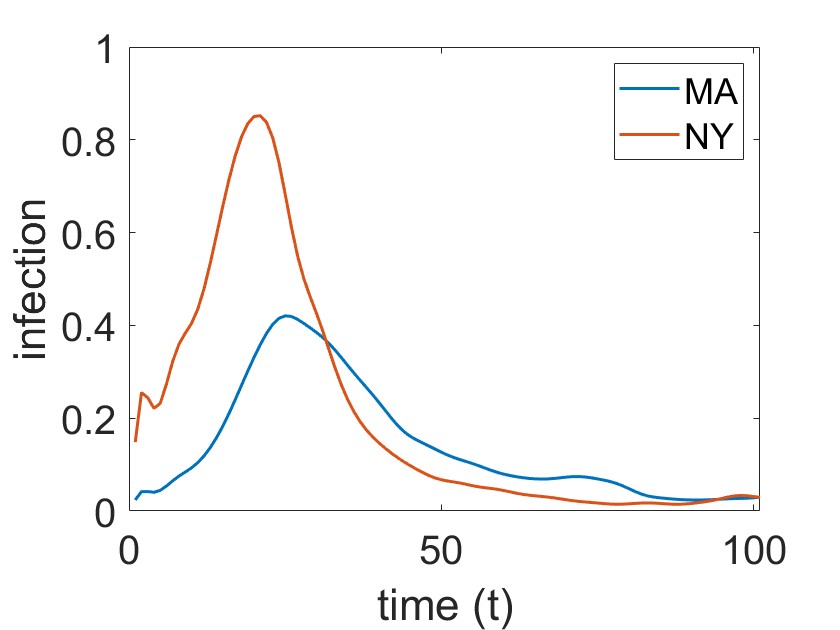}
    \end{subfigure}
    \qquad\qquad
    \begin{subfigure}{0.36\textwidth}
        \caption{Predicted mobility $f(m)$}
        \includegraphics[width=\hsize]{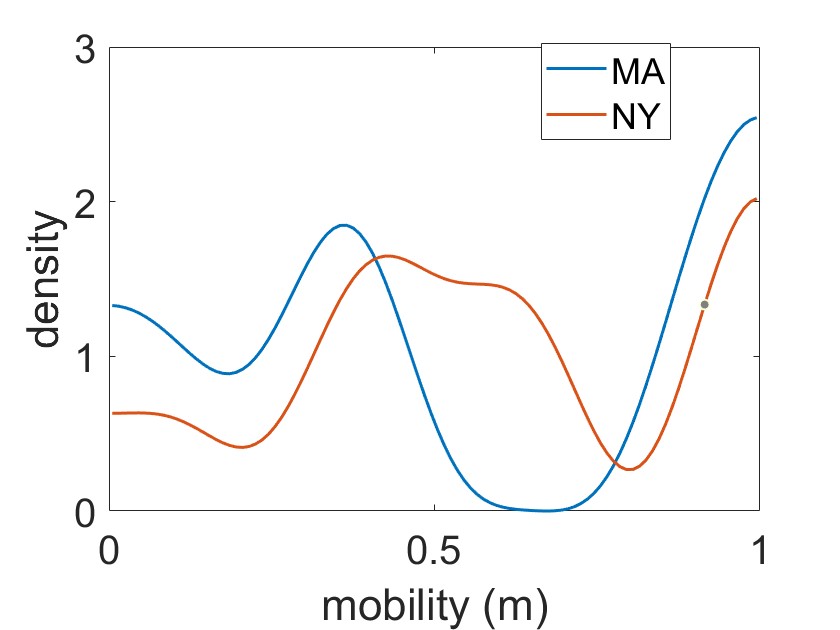}
    \end{subfigure}
    \caption{(a) The COVID-19 infection ratio (scaled by $50$) during the initial wave in Massachusetts and New York from February 29th, 2020, to September 16th, 2020. (b) Inferred mobility distributions from the data in (a). }
    \label{fig: real-data-learned}
\end{figure}

\section{Conclusion and outlook} \label{sec: conclusion}
We introduce and analyze density-based compartmental models in epidemiology that incorporate population heterogeneity through mobility. We assume that each individual has a mobility variable ranging from 0 to 1 and that individuals with higher mobility are more likely to transmit infectious diseases. We describe the disease dynamics by modeling the evolution of mobility distribution functions of the three compartments: susceptible, infected, and recovered. Notably, the classical SIR and SIRS models are special cases of our model when the mobility distributions are Dirac delta functions.

We derive key epidemiological parameters for our mobility-based compartmental models, including the basic reproduction number $\mathcal{R}_0$, the final pandemic size, and mean mobility. Our models effectively address the overestimation issue common in many homogeneous compartmental models when estimating the final pandemic size. 
Early in a pandemic, individuals with higher mobility tend to become infected more rapidly, resulting in a higher initial growth rate of cases compared to a homogeneous population.
As the pandemic evolves, the recovered high-mobility population imparts immunity to the broader population, particularly those with lower mobility. Most existing models neglect this mechanism, while ours captures it, offering a more accurate estimate of the basic reproduction number and the final pandemic size. We prove that the classical SIR model, where the entire population has the same mobility (\ie mobility distributions are Dirac delta functions), yields the largest final pandemic size compared to models with other mobility distributions. Numerical results show that a polarized mobility distribution results in a significantly smaller final pandemic size than the classical SIR model.

As it is often challenging to measure mobility distributions in practice, we formulate an inference problem to estimate the mobility distribution from the ratio of the total infected population. We prove that there is a one-to-one correspondence between the time series of the total infected population and the mobility distribution, which implies that the mobility distribution can be uniquely identified from the infection data. Additionally, we propose a numerical approach using a machine learning framework to infer the mobility distribution from the data and validate our results with real COVID-19 data.

In this paper, we incorporate heterogeneous mobility distributions into the SIR and SIRS models. It is also valuable to integrate mobility aspects into more complex compartmental models, such as those with additional compartments or those accounting for birth and death. Currently, we assume that each individual's mobility is time-independent, but one can relax this assumption to consider mobility dynamics influenced by social opinions, public policies, overall infection cases, and many other factors. Another potentially interesting extension is to combine the mobility-based compartmental models with control problems and explore how managing population mobility can reduce disease spread.

\section*{Acknowledgement} YL and NJ are partially supported by NSF DMS-2108628.

\bibliographystyle{abbrv}
\bibliography{references}

\end{document}